\documentclass[11pt]{article}

\usepackage{fullpage}
\usepackage{times}
\usepackage{soul}
\usepackage{url}
\usepackage[utf8]{inputenc}
\usepackage[small]{caption}
\usepackage{graphicx}
\usepackage{amsmath}
\usepackage{booktabs}

\usepackage[breaklinks=true]{hyperref}
\usepackage[svgnames]{xcolor}
\usepackage[capitalise,nameinlink]{cleveref}
\hypersetup{colorlinks={true},linkcolor={DarkBlue},citecolor=[named]{DarkGreen}}

\usepackage{amsmath,amssymb}
\usepackage{amsthm}
\usepackage{booktabs}

\usepackage{algorithm}
\usepackage{algpseudocode}

\usepackage{natbib}

\makeatletter
\renewcommand{\ALG@name}{Mechanism}
\makeatother

\usepackage{tikz}

\newcommand{\SW}{\text{SW}}

\newcommand{\boldv}{\mathbf{v}}

\newcommand{\sucv}{\boldsymbol{\succ}_{\mathbf{v}}}

\newcommand{\dist}{\text{dist}}

\newcommand{\bv}{\mathbf{v}}

\DeclareMathOperator*{\argmax}{arg\,max}

\newtheorem{claim}{Claim}

\newcommand{\vv}{\mathbf{v}}

\usepackage{epigraph}

\usepackage{bm}

\newtheorem{theorem}{Theorem}
\newtheorem{corollary}[theorem]{Corollary}

\newtheorem{innercustomthm}{Theorem}
\newenvironment{customthm}[1]
  {\renewcommand\theinnercustomthm{#1}\innercustomthm}
  {\endinnercustomthm}

\theoremstyle{definition}
\newtheorem{definition}{Definition}

\makeatletter

\usepackage{authblk}

\title{\bf Don't Roll the Dice, Ask Twice: \\ The Two-Query Distortion of Matching Problems and Beyond
}

\author[1]{Georgios Amanatidis}
\author[2]{Georgios Birmpas}
\author[3]{Aris Filos-Ratsikas}
\author[4]{Alexandros A. Voudouris}

\affil[1]{\emph{\normalsize Department of Mathematical Sciences, University of Essex, UK}}
\affil[2]{\emph{\normalsize Department of Computer, Control \& Management Engineering, Sapienza University of Rome, Italy}}
\affil[3]{\emph{\normalsize School of Informatics, University of Edinburgh, UK}}
\affil[4]{\emph{\normalsize School of Computer Science and Electronic Engineering, University of Essex, UK}}

\date{}

\begin{document}

\allowdisplaybreaks

\maketitle

\begin{abstract}
In most social choice settings, the participating agents express their preferences over the different alternatives in the form of linear orderings. While this clearly simplifies preference elicitation, it inevitably leads to poor performance with respect to optimizing a cardinal objective, such as the social welfare, since the values of the agents remain virtually unknown. This loss in performance because of lack of information is measured by the notion of {\em distortion}.  A recent array of works put forward the agenda of designing mechanisms that learn the values of the agents for a small number of alternatives via {\em queries}, and use this limited extra information to make better-informed decisions, thus improving distortion. Following this agenda, in this work we focus on a class of combinatorial problems that includes most well-known matching problems and several of their generalizations, such as One-Sided Matching, Two-Sided Matching, General Graph Matching, and $k$-Constrained Resource Allocation. We design {\em two-query} mechanisms that achieve the best-possible worst-case distortion in terms of social welfare, and outperform the best-possible expected distortion achieved by randomized ordinal mechanisms.
\end{abstract}

\section{Introduction} \label{sec:intro}
The notion of  \emph{distortion} in social choice settings was defined to capture the loss in aggregate objectives due to the lack of precise information about the preferences of the participants~\citep{procaccia2006distortion}. 
More concretely, the distortion was originally defined as a measure of the deterioration of the total happiness of the agents when access is  given only to the (ordinal) preference rankings of the agents, rather than to the complete numerical (cardinal) information about their preferences. 
This research agenda has successfully been applied to a plethora of different settings, giving rise to a rich and vibrant line of work in major venues at the intersection of computer science and economics. For a comprehensive overview, see the survey of \citet{survey2021}.

Out of all of these scenarios, some of the most fundamental are \emph{matching} problems, in which agents are matched to items or other agents, aiming to maximize the \emph{social welfare} of the matching (the total value of the agents). An example is the classic \emph{One-Sided Matching} setting \citep{hylland1979efficient}, where the goal is to match $n$ items to $n$ agents based on the preferences of the agents over the items. For this setting, \cite{Aris14} showed that the best achievable distortion is $\Theta(\sqrt{n})$. Importantly, this guarantee is only possible if one is allowed to use \emph{randomization} and the values of the agents are \emph{normalized}.\footnote{Note that if any of these assumptions is relaxed, it is impossible to achieve sub-linear distortion using only ordinal information.}

Moving on from merely preference rankings, \citet{amanatidis2020peeking} recently put forward the agenda of studying the tradeoffs between information and efficiency, when the employed mechanisms are equipped with the capability of learning the values of the agents via \emph{queries}. The rationale is that asking the agents for more detailed information about only \emph{a few} options is still cognitively not too burdensome, and could result in notable improvements on the distortion. This was indeed confirmed in that work for general social choice, and in a follow-up work for several matching problems \citep{Amanatidis2021matching}. Specifically, the latter work shows that it is possible to obtain distortion $O(n^{1/k})$ with $O(\log n)$ queries per agent for any constant integer $k$, and distortion $O(1)$ with $O(\log^2 n)$ queries per agent. Crucially, the mechanisms achieving these bounds {do not use randomization nor demand the values to be normalized}.

While these works make a significant first step, they leave some important questions unanswered. The mechanisms they propose require a logarithmic number of queries to achieve \emph{any} significant improvement. Answering that many queries might still be cognitively too demanding for the agents, especially when there is a large number of possible options. The main high-level motivation of this research agenda is that a small amount of information can be more valuable than randomization. But what does really constitute a ``small amount''? Ideally, we would like to design mechanisms that make only a few queries per agent, independently of the size of the input parameters. Since with a single query, sub-linear distortion bounds are not possible \citep{Amanatidis2021matching,amanatidis2020peeking}, the first fundamental question that we would like to answer is the following:

\begin{quote}
\emph{What is the best achievable distortion when we can only ask \emph{two} queries per agent?}
\end{quote}

\subsection{Results and Technical Overview} \label{sec:contribution}
We settle the aforementioned question for several matching problems, including \emph{One-Sided Matching}, \emph{Two-Sided Matching}, \emph{General Graph Matching}, and other more general graph-theoretic problems. 
For all matching problems considered, we show that there is a deterministic mechanism that makes two queries per agent, runs  in polynomial time, and achieves a distortion of $O(\sqrt{n})$. This upper bound is based on a novel mechanism, which we call \textsc{Match-TwoQueries} in the case of One-Sided Matching (see Mechanism~\ref{alg:Match-TwoQueries}). The mechanism asks two queries per agent and computes a maximum-weight matching based of the revealed values due to these queries. It starts by querying the agents at the first position of their preference rankings. For the second query, it computes a certain type of assignment $A$ of agents to items (or agents to agents in more general matching problems), to which we refer as a \emph{sufficiently representative assignment}, and queries the agents about the items they are assigned to in $A$. The existence of such an assignment for all instances is far from trivial, and one of our main technical contributions is to show its existence and efficient computation for the wide range of problems we consider.  


We also show that no deterministic mechanism for these settings that makes two queries per agent can achieve a distortion better than $\Omega(\sqrt{n})$. This lower bound follows by a more general construction yielding a lower bound of $\Omega(n^{1/\lambda})$ on the distortion of any mechanism that makes a constant number $\lambda$ of queries for any of these mechanisms. This mirrors the corresponding lower bounds of \cite{Amanatidis2021matching} for One-Sided Matching.

While our results apply to general matching settings, their most impressive implications are for One-Sided Matching: 
We show that by using only \emph{two} cardinal queries per agent, we can match the bound of $\Theta(\sqrt{n})$ for purely ordinal mechanisms, \emph{without requiring randomization or any normalization}. \textsc{Match-TwoQueries} clearly also outperforms another mechanism of \cite{Amanatidis2021matching}, which uses two queries and achieves a distortion of ${O}(n^{2/3} \sqrt{\log{n}})$ assuming that the values of each agent sum up to $1$. In contrast, our mechanism works for unrestricted values, and achieves the best possible distortion of $O(\sqrt{n})$ based on conceptually much simpler ideas. 

\paragraph{Results for general social choice.} 
Given that our approach works for a wide variety of matching problems, one might be curious as to whether similar arguments could be used to show bounds for the general social choice setting, where $n$ agents have preferences over $m$ alternatives, and the goal is to select an alternative with high social welfare; this was after all the original setting that \citet{amanatidis2020peeking} studied in the introduction of the information-distortion tradeoff research agenda. In this setting, the situation is quite similar: the upper bounds follow by mechanisms that ask $O(\log m)$ queries, and nothing positive is known for smaller numbers of queries. 

We show that a mechanism with structure similar to that of \textsc{Match-TwoQueries} can indeed achieve a distortion of $O(\sqrt{m})$ using only two queries, subject to being able to compute a sufficiently representative set of alternatives, which is analogous of the sufficiently representative assignment in matching problems. It turns out that this property is very closely connected to the notion of an \emph{(approximately) stable committee} \citep{jiang2020approximately,cheng2020group}, and it follows that it exists when $m = \Omega(n)$, thus allowing us to obtain the desired bound of $O(\sqrt{m})$ when this is true. This case is quite natural, as it captures instances where a group of people need to decide over a large set of possible options (e.g., shortlisting candidates for a job, deciding the best paper for a conference, etc.). Interestingly, in contrast to the matching setting for which we show that sufficiently representative assignments can be found via a simple greedy algorithm, computing sufficiently representative sets of alternatives in general social choice requires rather involved techniques~\citep{jiang2020approximately,cheng2020group}. An obvious open question here is whether the $O(\sqrt{m})$ bound can also be achieved by asking only two queries when $m=o(n)$. This seems to be a more challenging task to prove; we discuss it further in Section~\ref{sec:open}.


We also show that the bound of $O(\sqrt{m})$ is the best possible, as part of a more general distortion lower bound of $\Omega(m^{1/\lambda})$ for mechanisms that make a constant number $\lambda$ of queries per agent; the latter result significantly improves the previously known lower bound of $\Omega(m^{1/2(\lambda+1)})$ \citep{amanatidis2020peeking}. 


\paragraph{Roadmap.}
For the sake of presentation, we fully demonstrate how our methodology works for the {\em One-Sided Matching} problem in Section~\ref{sec:matching-2-queries}. Before doing so, we start with some necessary notation and terminology in Section~\ref{sec:prelim}. In Section~\ref{sec:generalizations}, we briefly discuss other graph-theoretic problems for which our methodology can be applied.
Our results for general social choice are presented in Section~\ref{sec:social-choice}. We conclude with some interesting open problems in Section~\ref{sec:open}.

\subsection{Additional Related Work}
The literature on the distortion of ordinal mechanisms in social choice is long and extensive, focusing primarily on settings with normalized utilities (e.g., \citep{boutilier2015optimal,ebadian2022optimized,caragiannis2017subset,ratsikas2020distributed}), or with metric preferences (e.g., \citep{anshelevich2018approximating,anshelevich2017randomized,anshelevich2022distributed,CSV22,charikar2022randomized,gkatzelis2020resolving}); see the survey of \citet{survey2021} for a detailed exposition. 
The distortion of mechanisms for One-Sided Matching and more general graph-theoretic problems has been studied in a series of works for a variety of preference models, but solely with ordinal information \citep{anshelevich2018ordinal,abramowitz2017utilitarians,anshelevich2017tradeoffs,anshelevich2016blind,Aris14,caragiannis2016truthful}.

Besides the papers of \citet{amanatidis2020peeking,Amanatidis2021matching}, the effect of limited cardinal information on the distortion has also been studied in other works~\citep{abramowitz2019awareness,mandalefficient,mandal2020optimal,benade2017preference}. Mostly related to us is the paper of \citet{ma2021matching} which considered the One-Sided Matching problem with a different type of cardinal queries, and showed qualitatively similar results to \cite{Amanatidis2021matching} for Pareto optimality (rather than social welfare).

Our upper bound for the general social choice setting makes use of the results of \citet{cheng2020group} and \citet{jiang2020approximately} for (approximately) stable committees (see also \citet{aziz2017condorcet}); a stable committee is very similar to a representative set of alternatives in our terminology. \citet{cheng2020group} showed that, while \emph{exactly} stable committees do not always exist \citep{jiang2020approximately}, finding a random version of such committees, coined \emph{stable lotteries}, is always possible and can be done in polynomial time. Later on, \citet{jiang2020approximately} showed that, via an intricate derandomization process, stable lotteries can yield approximately stable committees, where the approximation is a small multiplicative constant; for our purposes, this is sufficient. Interestingly, very recently, \citet{ebadian2022optimized} used stable lotteries to construct a purely ordinal randomized social choice mechanism that achieves the best possible distortion under unit-sum normalized values.

\section{Preliminaries on One-Sided Matching, Mechanisms, and Distortion} \label{sec:prelim}
In One-Sided Matching, there is a set $\mathcal{N}$ of $n$ {\em agents} and a set $\mathcal{A}$ of $n$ {\em items}. Each agent $i \in \mathcal{N}$ has a {\em value} $v_{i,j}$ for each item $j \in \mathcal{A}$; we refer to the matrix $\boldv=(v_{i,j})_{i \in \mathcal{N}, j \in \mathcal{A}}$ as the {\em valuation profile}. A {\em (one-sided) matching} $X: \mathcal{N} \rightarrow \mathcal{A}$ is a bijection from the set of agents to the set of items, i.e.,  each agent is \emph{matched} to a different single  item. Our goal is to choose a matching $X$ to maximize the {\em social welfare}, defined as the total value of the agents for the items they have been matched to according to $X$: 
$\SW(X| \boldv) = \sum_{i \in \mathcal{N}} v_{i,X(i)}.$
Usually $\boldv$ is clear from the context, so we then simplify our notation to $\SW(X)$ for the social welfare of matching $X$.

As in most of the related literature, we assume that we do not have access to the valuation profile of the agents. Instead, we have access to the {\em ordinal preference} $\succ_i$ of each agent $i$, which is derived from the values of the agent for the items, such that $a \succ_i b$ if $v_{i,a} \geq v_{i,b}$; we refer to the vector $\sucv = (\succ_i)_{i \in \mathcal{N}}$ as the {\em ordinal profile} of the agents.

\smallskip

A {\em mechanism} $\mathcal{M}$ in our setting operates as follows: 
It takes as input the ordinal profile $\sucv$ of the agents.
It then makes a number $\lambda \geq 1$ of {\em queries} per agent to learn part of the valuation profile. In particular, each agent is asked her value for at most $\lambda$ items. 
Given the answers to the queries, and also using the ordinal profile, $\mathcal{M}$ computes a feasible solution (here a matching) $\mathcal{M}(\sucv)$. 

In this paper we focus on mechanisms that make two queries per agent, i.e., $\lambda=2$, and compute a solution of high social welfare. However, pinpointing an (approximately) optimal solution without having full access to the valuation profile of the agents can be quite challenging; the ordinal profile may be consistent with a huge number of different valuation profiles, even after the queries. Nevertheless, we aim to achieve the best asymptotic performance possible, as quantified by the notion of {\em distortion}.

\begin{definition}\label{def:distortion}
The {\em distortion} of a mechanism $\mathcal{M}$ is the worst-case ratio (over the set $\mathcal{V}$ of all valuation profiles in instances with $n$ agents and $n$ items) 
between the optimal social welfare and the social welfare of the solution chosen by $\mathcal{M}$:
\begin{align*}
\dist(\mathcal{M}) = \sup_{\boldv \in \mathcal{V}, |\mathcal{N}|=n, |\mathcal{A}|=n} \frac{\max_{X \in \mathcal{X}} \SW(X | \boldv)}{ \SW(\mathcal{M}(\sucv) | \boldv)},
\end{align*}
where $\mathcal{X}$ is the set of all matchings between $\mathcal{N}$ and $\mathcal{A}$.
\end{definition}

\section{An Optimal Two-Query Mechanism} \label{sec:matching-2-queries}
In this section, we present a mechanism for One-Sided Matching that makes two queries per agent and achieves a distortion of $O(\sqrt n)$. 
Due to the lower bound of $\Omega(n^{1/\lambda})$ on the distortion of any mechanism that can make up to $\lambda$ queries per agent shown by \citet{Amanatidis2021matching}, our mechanism is asymptotically best possible when $\lambda=2$. 

Without any normalization assumptions about the valuation functions, it is easy to see that a mechanism cannot have \emph{any} guarantee unless it queries every agent about her favorite item. However, there are no obvious criteria suggesting how to use the \emph{second} query. Before we present the details of our mechanism, we define a particular type of assignment of agents to items that will be critical for deciding where to make the {second} queries. 

\begin{definition} \label{def:nia}
A many-to-one assignment $A$ of agents to items (i.e., each agent is assigned to one item, but multiple agents may be assigned to the same item) is a \emph{sufficiently representative assignment} if
(a) For every item $j \in \mathcal{A}$, there are at most $\sqrt{n}$ agents assigned to $j$;
(b) For any matching $X$, there are at most $\sqrt{n}$ agents that prefer the item they are matched to in $X$ to the item they are assigned to in $A$.
\end{definition}
A natural question at this point is whether a sufficiently representative assignment exists for any instance, and if so, whether it can be efficiently computed. In Section \ref{sec:NIA}, we present a simple polynomial-time algorithm for this task.

\subsection{The Mechanism}\label{sec:TheMechanism}
Our mechanism \textsc{\sc Match-TwoQueries} (Mechanism~\ref{alg:Match-TwoQueries}) first queries every agent  about her  favorite item. Next, it computes a sufficiently representative assignment $A$ (see Section \ref{sec:NIA}) and queries each agent about the item she is assigned to in $A$. Finally, it outputs a matching that maximizes the social welfare based \emph{only} on the revealed values (all other values are set to $0$).
Although computational efficiency is not our primary focus here, if we use a polynomial-time algorithm for computing a maximum weight matching (e.g., the Hungarian method \citep{Kuhn56}), \textsc{\sc Match-TwoQueries} runs in polynomial time as well. 

\begin{algorithm}[!ht]
{\fontsize{10}{12}\selectfont
\begin{algorithmic}[1]
	\State Query each $i\in \mathcal{N}$ about her favorite item w.r.t.~$\succ_i$
	\State Compute a \emph{sufficiently representative assignment} $A$
	\State Query each agent about the item she is assigned to in $A$
	\State Set all non-revealed values to $0$
	\State \Return a maximum-weight perfect matching $Y$
\end{algorithmic}
}
\caption{\textsc{\sc Match-TwoQueries}$(\mathcal{N}, \mathcal{A}, \sucv)$} \label{alg:Match-TwoQueries}
\end{algorithm} 

Of course, if we compare the mechanism's behaviour to an actual optimal matching $X$, we expect to see that we asked agents about the ``wrong'' items most of the time: for many agents the second query is about better items than what they are matched to in $X$, and for many agents it is about worse items.  
The desired bound of $O(\sqrt{n})$ on the distortion of {\sc Match-TwoQueries} is established by balancing the loss due to each of these two cases.

\begin{theorem}\label{thm:two-queries-matching}
{\sc Match-TwoQueries} has distortion $O(\sqrt{n})$.
\end{theorem}

\begin{proof}
Consider any instance with valuation profile $\bv$. 
Let $Y$ be the matching computed by the {\sc Match-TwoQueries} mechanism when given as input the ordinal profile $\sucv$, and let $X$ be an optimal matching. Let $\SW_R(Y)$ be the \emph{revealed} social welfare of $Y$, i.e., the total value of the agents for the items they are matched to in $Y$ and for which they {\em were queried} about. We will show that $\SW(X) \leq (1 + 2\sqrt{n} ) \cdot \SW_R(Y)$, and then use the fact that $\SW(Y) \geq \SW_R(Y)$ to directly get the desired bound on the distortion.

We can write the optimal social welfare as
\[
\SW(X) = \SW_R(X) + \SW_C(X),
\]
where $\SW_{R}(X)$ is the revealed social welfare of $X$ that takes into consideration only the values revealed by the queries, whereas $\SW_{C}(X)$ is the \emph{concealed} social welfare of $X$ that takes into consideration only the values not revealed by any queries.
Since $Y$ is the matching that maximizes the social welfare based only on the revealed values, we have that 
\begin{align}\label{eq:matching-revealed}
\SW_R(X) \leq \SW_R(Y).
\end{align}
To bound the quantity $\SW_C(X)$, let $S$ be the set of agents who are not queried about the items they are matched to in $X$. We partition $S$ into the following two subsets consisting of agents for whom the second query of the mechanism is used to ask about items that the agents consider {\em better} or {\em worse} than the items they are matched to in $X$. Recall that an agent $i$ is queried about the item $A(i)$ she is assigned to according to the sufficiently representative assignment $A$. So, $S$ is partitioned into
\begin{equation*}
\begin{aligned}
S^{\geq} &= \left\{ i \in S: v_{i,A(i)} \geq v_{i,X(i)}\right\}, \text{\quad and \quad}
S^{<} &= \left\{ i \in S: v_{i,A(i)} < v_{i,X(i)}\right\}.
\end{aligned}
\end{equation*}
Given these sets, we can now write
\[\SW_C(X) = \SW_C^{\geq}(X) + \SW_C^{<}(X),\]
where
\[\SW_C^{\geq}(X) = \sum_{i \in S^{\geq}}v_{i,X(i)}\] and
\[\SW_C^{<}(X) = \sum_{i \in S^{<}}v_{i,X(i)}.\]

For every item $j$, let $S_j^{\geq} = \{ i \in S^{\geq}: A(i)=j \}$ be the set of all agents in $S^{\geq}$ that are queried about $j$ by the mechanism using the second query. Thus, $S^{\geq} = \bigcup_{j \in \mathcal{A}}S_j^{\geq}$.
Since $A$ is a sufficiently representative assignment, $|S_j^{\geq}| \leq \sqrt{n}$ for every item $j$. Therefore, 
\begin{align} \label{eq:matching-concealed-geq}
\SW_C^{\geq}(X) 
&= \sum_{j \in \mathcal{A}} \sum_{i \in S_j^{\geq}} v_{i,X(i)} 
\leq \sum_{j \in \mathcal{A}} \sum_{i \in S_j^{\geq}} v_{i,j} \nonumber\\ 
&\leq \sum_{j\in \mathcal{A}} |S_j^{\geq}| \cdot \max_{i \in S_j^{\geq}} v_{i,j}   
\leq \sqrt{n} \sum_{j\in \mathcal{A}} \max_{i \in S_j^{\geq}} v_{i,j}  \nonumber\\ 
& \leq \sqrt{n} \cdot \SW_R(Y).
\end{align}
For the last inequality, recall that $A$ assigns every agent to a single item, and thus the sets $S_j^{\geq}$ are disjoint.
In addition, the values of all the agents for the items they are matched to according to $A$ are revealed by the second query of the mechanism. Since we can match the agent in $S_j^{\geq}$ that has the maximum value for $j$ to $j$, and $Y$ maximizes the social welfare based on the revealed values, we obtain that $\SW_R(Y) \geq \sum_{j\in \mathcal{A}} \max_{i \in S_j^{\geq}} v_{i,j}$.

Next consider the quantity $\SW_C^{<}(X)$. By the fact that $A$ is a sufficiently representative assignment, it follows that $|S^{<}| \leq \sqrt{n}$; otherwise $X$ would constitute a matching for which there are strictly more than $\sqrt{n}$ agents that prefer the  item they are matched to in $X$ to the item they are assigned to by $A$. Combined with the fact that all agents are queried at the first position of their ordinal preferences, we obtain
\begin{align} \label{eq:matching-concealed-<} 
\SW_C^{<}(X) 
&= \sum_{i \in S^{<}} v_{i,X(i)} 
\leq \sum_{i \in S^{<}} \max_{j \in \mathcal{A}} v_{i,j} \nonumber \\
&\leq |S^{<}| \cdot \max_{i \in S^{<}} \max_{j \in \mathcal{A}} v_{i,j} \nonumber \\
&\leq \sqrt{n} \cdot \SW_R(Y).
\end{align}

The bound follows directly by \eqref{eq:matching-revealed}, \eqref{eq:matching-concealed-geq} and \eqref{eq:matching-concealed-<}.
\end{proof}

\subsection{Computing Sufficiently Representative Assignments}\label{sec:NIA}

To establish the correctness of {\sc Match-TwoQueries}, we need to ensure that a sufficiently representative assignment exists for any ordinal profile and that it can be computed efficiently. 
For this we present a simple polynomial time algorithm, which we call {\sc $\sqrt{n}$-Serial Dictatorship} (Mechanism~\ref{alg:SD-sqrt}). This algorithm creates $\sqrt{n}$ copies of each item and then runs a serial dictatorship algorithm, which first fixes an ordering of the agents and then assigns each agent to her most preferred available item according to her ordinal preference. It is easy to see that the running time of {\sc $\sqrt{n}$-Serial Dictatorship} is polynomial (in particular, it is $O(n^{1.5})$).

\begin{algorithm}[h!]
{\fontsize{10}{12}\selectfont
\begin{algorithmic}[1]
 \State Let $\mathcal{B}$ be a multiset containing $\sqrt{n}$ copies of each  $j\in\mathcal{A}$ \;
  \For {\normalfont every agent $i \in \mathcal{N}$}
	   \State Let $\alpha_i$ be a most preferred item of agent $i$ in $\mathcal{B}$ \;
	   \State Remove $\alpha_i$ from $\mathcal{B}$ \;
   \EndFor\label{blah}
\State \Return $A = (\alpha_i)_{i \in \mathcal{N}}$
\end{algorithmic}
}
\caption{\textsc{\sc $\sqrt{n}$-Serial Dictatorship}$(\mathcal{N}, \mathcal{A}, \sucv)$} 
\label{alg:SD-sqrt}
\end{algorithm}

\begin{theorem}\label{thm:matching-GAGA}
For any instance, the output of {\sc $\sqrt{n}$-Serial Dictatorship} is a sufficiently representative assignment.
\end{theorem}

\begin{proof}
Let $A$ be the output of the algorithm. During the execution of the algorithm, whenever every copy of an item has been assigned, we say that such an item is \emph{exhausted}. Assume, towards a contradiction, that $A$ is not a sufficiently representative assignment. By construction, every item is assigned to at most $\sqrt{n}$ agents, so there must be a matching violating the second condition of Definition \ref{def:nia}. That is, there is 
a subset of items $\mathcal{A}'$ and a subset of agents $\mathcal{N}'$, such that
 $|\mathcal{A}'| = |\mathcal{N}'| > \sqrt{n}$, and
each agent $i \in \mathcal{N}'$ prefers to be assigned to a distinct item $\beta_i \in \mathcal{A}'$ (i.e., $\beta_i\neq \beta_j$ for $i \neq j$) instead of the item she is assigned to in $A$. 

Consider any agent $i \in \mathcal{N}'$. 
The fact that this agent was not assigned to $\beta_i$ by the algorithm implies that when the agent was picked, item $\beta_i$ was exhausted. Since this is true for all agents in $\mathcal{N}'$, at the end of the algorithm all items of $\mathcal{A}'$ must be exhausted. However, an item is exhausted when all its $\sqrt{n}$ copies have been assigned and there are $n$ agents in total, so we can only have as many as ${n}/{\sqrt{n}} = \sqrt{n}$ exhausted items. This  means that $|\mathcal{A}'| \leq \sqrt{n}$, a contradiction.
\end{proof}

\section{Further Combinatorial Optimization Problems}\label{sec:generalizations}

The approach of Section \ref{sec:matching-2-queries} can be extended to a much broader class of graph-theoretic problems. 
Informally, our approach works when the objective is to maximize an additive function over subgraphs of a given graph which contain all ``small'' matchings and have constant maximum degree.
We make the space of feasible solutions more precise in the following definition. 
\begin{definition}\label{def:feasible_families}
Given a constant $k\in \mathbb{N}$ and a 
weighted graph $G$ on $n$ nodes, we say that a family $\mathcal{F}$ of subgraphs of $G$ is a \emph{matching extending $k$-family} if:
\begin{itemize}
    \item Graphs in $\mathcal{F}$ have maximum 
    degree at most $k$;
    \item For any matching $M$ of $G$ of size at most $\lfloor n/3k \rfloor$, there is a graph in $\mathcal{F}$ containing $M$.
\end{itemize}
\end{definition}
Clearly, the set of matchings of a graph (viewed as subgraphs rather than subsets of edges) is a matching extending $1$-family, but it is not hard to see that Definition \ref{def:feasible_families} captures  other constraints, like subgraphs that are unions of disjoint paths and cycles (matching extending $3$-family) or unions of disjoint cliques of size $k$ (matching extending $(k-1)$-family). 

We are ready to introduce the general full information optimization problem that we tackle here; we then move on to its social choice analog. 
As this is a special case of the class of problems captured by Max-on-Graphs (introduced by \citet{Amanatidis2021matching}), we use a similar formulation and name.
Note that, in the above definition, the family $\mathcal{F}$ is independent of the weights $w$. This is necessary as $w$ will be unknown in general.



\medskip
\noindent 
\textbf{\em $\bm{k}$-Max-on-Graphs}: Given a constant $k\in \mathbb{N}$, a 
weighted graph $G = (U, E, w)$, and a concise description of a matching extending $k$-family $\mathcal{F}$, find a solution $H^* \in \allowbreak \arg\max_{H\in\mathcal{F}} \sum_{e\in E(H)} w(e)$.

\medskip


\noindent One-Sided Matching, as studied in Section \ref{sec:matching-2-queries}, is the special case of $k$-Max-On-Graphs, where $G$ is the complete bipartite graph on the set of agents $\mathcal{N}$ and the set of items $\mathcal{A}$, the weight of an edge $\{i,j\}$ is the value $v_{i, j}$ of agent $i$ for item $j$, and $\mathcal{F}$ contains all the $1$-factors of $G$. Note that the weights of the graph in this case are defined in terms of the valuation functions of the agents. Moreover, recall that in our setting only the ordinal preferences of the agents are given and their cardinal values can be accessed only via queries; so, we do not know the weights that have not been revealed by a query. This is the case for all the problems we are interested in, and is captured by the next definition. To avoid unnecessary notation, items are modeled as dummy agents with all their cardinal values equal to $0$. In addition, we write $w(H) := \sum_{e\in E(H)} w(e)$.

\medskip

\noindent\textbf{\em Ordinal-$\bm{k}$-Max-on-Graphs}:
Fix a constant $k \in \mathbb{N}$ and let $\mathcal{N}$ be a set of $n$ agents. A 
weighted graph $G = (\mathbb{N}, E, w)$ is given \emph{without} its weights.
Every agent $i\in \mathcal{N}$ has a (private) valuation function $v_i: \mathbb{N} \rightarrow \mathbb{R}_{\geq 0}$, such that, for every $e = \{i,j\} \in E$,
\[
w(e) =
v_i(j) + v_j(i). 
\]
\noindent
We are also given an {\em ordinal profile} $\succ_\vv = (\succ_i)_{i \in \mathcal{N}}$ that is consistent to $\vv = (v_i)_{i \in \mathcal{N}}$, and
a concise description of a matching extending $k$-family $\mathcal{F}$.
The goal is to find $H^* \in \allowbreak \arg\max_{H\in\mathcal{F}} w(H)$.


Besides One-Sided Matching, a large number of problems that are relevant to computational social choice are captured by Ordinal-$k$-Max-on-Graphs. We give a few examples: 

\medskip


\noindent\emph{\textbf{General Graph Matching}}: Given a weighted graph $G = (U, E, w)$, find a matching of maximum weight, i.e., $\mathcal{F}$ contains the matchings of $G$ and clearly is a matching extending $1$-family.
In the social choice analog of the problem, $U=\mathcal{N}$. 

\medskip

\noindent\emph{\textbf{Two-Sided Matching}}:
This is a special case of General Graph Matching in which $G = (U_1\cup U_2, E, w)$ is a bipartite graph. It is an extensively studied problem in economics and computational social choice \citep{GaleShapley62,RothSotomayor92}.



\medskip


\noindent\emph{\textbf{$\bm{k}$-Clique Packing}}: Given a weighted complete graph $G = (U, E, w)$, where $|U| = n$ is a multiple of $k\ge 2$, the goal is to partition $U$ into $|U|/k$ clusters of size $k$ to maximize the total weight of the edges in the clusters. That is, $\mathcal{F}$ contains all spanning subgraphs of $G$ that are the union of cliques of size $k$. As claimed above, $\mathcal{F}$ is a matching extending $(k-1)$-family: clearly every graph in $\mathcal{F}$ has maximum degree $k-1$, and any matching of size $\lfloor n/(3(k-1)) \rfloor$ (which is less than $n/k$) can be extended to a graph in $\mathcal{F}$ by arbitrarily grouping each pair of matched nodes with $k-2$ unmatched nodes, and then arbitrarily grouping the remaining nodes $k$ at a time.

This problem generalizes General Graph Matching (for which $k=2$) and often is referred to as Max $n/k$-Sum Clustering in the literature; see \citep{anshelevich2016blind}. In its social choice analog, $U=\mathcal{N}$. 

\medskip

\noindent\emph{\textbf{General Graph $\bm{k}$-Matching}}: Given a weighted graph $G = (U, E, w)$, find a $k$-matching of maximum weight, i.e., $\mathcal{F}$ contains all the subgraphs of $G$ where each node has degree at most $k$. As $\mathcal{F}$ already contains all matchings of $G$ of any size, it is straightforward that it is a matching extending $k$-family.
In the social choice analog of the problem, $U=\mathcal{N}$. 

\medskip

\noindent\emph{\textbf{$\bm{k}$-Constrained Resource Allocation}}: Given a bipartite weighted graph $G = (U_1\cup U_2, E, w)$, the goal is to assign at most $k$ nodes of $U_2$ to each node in $U_1$ so that the total weight of the corresponding edges is maximized. That is, $\mathcal{F}$ contains the subgraphs of $G$ where each node in $U_1$ has degree at most $k$ and each node in $U_2$ has degree at most $1$. Again, $\mathcal{F}$ already contains all matchings of $G$ of any size, so it is a matching extending $k$-family.

This problem generalizes One-Sided Matching. In its social choice analog, $\mathcal{N}=U_1 \cup U_2$ is partitioned into the ``actual agents'' $\mathcal{N}_1 = U_1$ and the ``items'' $\mathcal{N}_2 = U_2$, and 
$v_i(j)$ can be strictly positive only for $i \in \mathcal{N}_1, j\in \mathcal{N}_2$.

\medskip

\noindent\emph{\textbf{Short Cycle Packing}}: Given an integer $\ell$ and a weighted complete graph $G = (U, E, w)$, the goal is to find a collection of node-disjoint cycles of length at most $\ell$ so that their total weight is maximized.
Here, $\mathcal{F}$ contains any such collection of short cycles. Arguing as in ${k}$-Clique Packing, it is straightforward to see that $\mathcal{F}$ is a matching extending $(\ell-1)$-family (although it is not hard to show that it is actually a matching extending $3$-family).
The social choice analog of the problem has $U=\mathcal{N}$, and is closely related to Clearing Kidney $\ell$-Exchanges \citep{abraham2007clearing}.



\medskip

It is straightforward to extend the notion of distortion (Definition~\ref{def:distortion}) for Ordinal-$\bm{k}$-Max-on-Graphs by taking the supremum over all instances of a certain size $n$ and letting $\mathcal{X}$ be the set of feasible solutions of each instance.

As already discussed in the Introduction, for One-Sided Matching, \citet{Amanatidis2021matching} showed a lower bound of $\Omega(n^{1/\lambda})$ on the distortion of all deterministic mechanisms that can make up to $\lambda\geq 1$ queries per agent. Using this, we can get the analogous result for all the aforementioned problems. Although for some of them, like Two-Sided Matching and General Graph Matching, the lower bound is immediate, here we show it for \emph{any} problem captured by Ordinal-${k}$-Max-on-Graphs. For the statement of the theorem, $k\in \mathbb{N}$ is a constant, and we assume that for every graph $G$ a matching extending $k$-family $\mathcal{F}(G)$ is specified.


\begin{customthm}{3}
\label{thm:general-lower}
No deterministic mechanism using at most $\lambda \geq 1$ queries per agent can achieve a distortion better than $\Omega(n^{1/\lambda})$ for Ordinal-${k}$-Max-on-Graphs with feasible solutions given by $\mathcal{F}(\,\cdot\,)$.
\end{customthm}
\vspace{-7pt}
\begin{proof}
We are going to show that if we had a deterministic mechanism $\mathcal{M}$ for Ordinal-${k}$-Max-on-Graphs that makes at most $\lambda \geq 1$ queries per agent and achieves distortion $o(n^{1/\lambda})$, then we could design a deterministic mechanism for One-Sided Matching that also makes at most $\lambda$ queries per agent and has distortion $o(n^{1/\lambda})$. As the latter is impossible~\citep{Amanatidis2021matching}, that would imply that the lower bound applies to Ordinal-${k}$-Max-on-Graphs as well.



Let $(\mathcal{N}, \mathcal{A}, \sucv)$ be an arbitrary instance of One-Sided Matching with $|\mathcal{N}|=|\mathcal{A}|=n$ and underlying weights defined by $\boldv=(v_{i,j})_{i \in \mathcal{N}, j \in \mathcal{A}}$. We essentially use the same instance for Ordinal-${k}$-Max-on-Graphs: 
A complete bipartite weighted graph $G = (U_1\cup U_2, E, w)$ with $U_1=\mathcal{N}$, $U_2 = \mathcal{A}$, and valuation functions defined as $u_i(j) = v_{i, j}$ and  $u_j(i) = 0$ for every  $i \in \mathcal{N}, j \in \mathcal{A}$; the induced ordinal profile is also well-defined. 
Clearly, the matchings in the two instances are exactly the same and have the same weight (although they may not be feasible with respect to $\mathcal{F}(G)$). 
However, the feasible solutions for Ordinal-${k}$-Max-on-Graphs include subgraphs where the nodes may have degree up to $k$ instead of $1$. We need to relate the weight of an (approximately) optimal solution for the Ordinal-${k}$-Max-on-Graphs instance to the value of an optimal matching for the One-Sided Matching instance.

Let $M$ be a maximum weight matching in $G$ (and thus a maximum-value matching for the original One-Sided Matching instance) and $H\in \mathcal{F}(G)$ be an optimal solution. Consider the submatching $\hat{M}$ of $M$ that uses the $\lfloor 2n/3k \rfloor$ heaviest edges of $M$. 
Using the fact that $\lfloor x \rfloor \ge x/2$ for $x\ge 1$, we get
\begin{align} \label{eq:heavy-submatching}
w(\hat{M})
&\ge \frac{\lfloor 2n/3k \rfloor}{n} w(M)
\ge \frac{1}{3k} w(M).
\end{align}
Since $\mathcal{F}(G)$ is a matching extending $k$-family and $\hat{M}$ is sufficiently small (since $|V(G)|=2n$), there is some $\hat{H}\in \mathcal{F}(G)$ such that $\hat{M}$ is a subgraph of $\hat{H}$. As $H$ is a maximum-weight element of $\mathcal{F}(G)$, we directly get $w(H) \ge w(\hat{H}) \ge w(\hat{M})$ and, combining with \eqref{eq:heavy-submatching}, we have 
\begin{align} \label{eq:heavy-solution-H}
w(H) \ge \frac{1}{3k} w(M).
\end{align}

Now, if $H'\in \mathcal{F}(G)$ is an $\alpha$-approximate solution to the same Ordinal-${k}$-Max-on-Graphs instance, then \eqref{eq:heavy-solution-H} implies 
\begin{align} \label{eq:heavy-approx-solution}
w(H') \ge \frac{1}{\alpha} w(H)\ge \frac{1}{3\alpha k} w(M).
\end{align}
We can construct a matching from $H'$ using only ordinal information. 
In particular, for each $i\in U_1$, among the edges in $H'$ that are incident to $i$, we keep the best one with respect to $\succ_i$. 
Of course, the resulting graph $H''$ may not be a matching, as each node in $U_2$ may still have degree up to $k$.
However, note that this process also keeps at least a $1/k$ fraction of the weight incident to each $i\in U_1$, and thus of the total weight. So, \eqref{eq:heavy-approx-solution} implies
\begin{align*}
w(H'')\ge \frac{1}{k} w(H') \ge \frac{1}{3\alpha k^2} w(M).
\end{align*}
We repeat the process for the remaining nodes: for each $j\in U_2$, we keep the best of its edges in $H''$ with respect to $\succ_j$. 
Now the resulting graph $M'$ \emph{is} a matching and has at least a $1/k$ fraction of the total weight of $H''$ and, thus, 
\[w(M')\ge \frac{1}{k} w(H'') \ge \frac{1}{3\alpha k^3} w(M).\]

If needed, we can extend $M'$ to a perfect matching $M''$ by arbitrarily matching the unmatched nodes of $U_1$ and $U_2$, and consider its analog in the original instance. Clearly,   
$w(M'')\ge w(M')$, and thus $M''$ is a $(3 \alpha k^3)$-approximate solution for the original One-Sided Matching instance. Therefore, if there existed a mechanism $\mathcal{M}$ with distortion $\alpha = o(n^{1/\lambda})$ for Ordinal-${k}$-Max-on-Graphs with feasible solutions given by $\mathcal{F}$, we could use it for the above instance to get $H'$ and then $M''$, which would have weight within a factor of $o(3 k^3 n^{1/\lambda})$ from a maximum weight matching. Since $k$ is a constant, this would imply a distortion of $o(n^{1/\lambda})$ for One-Sided Matching, a contradiction.
\end{proof}

We are particularly interested in the case of $\lambda = 2$. In the next two sections we are going to present a mechanism for this case, which is asymptotically optimal, namely it achieves distortion $O(\sqrt{n})$, matching the lower bound we just derived.

\subsection{Sufficiently Representative Assignments}\label{sec:gaga}
We now revisit the notion of a \emph{sufficiently representative assignment}. We will appropriately adjust it to refer to a single set of agents (which, for the case of One-Sided Marching, includes both the actual agents and the items), and also incorporates the parameter $k$ from the definition of Ordinal-$k$-Max-on-Graphs. 

\begin{definition} \label{def:gaga}
Given $\mathcal{N}_1, \mathcal{N}_2 \subseteq \mathcal{N}$ and $k \in \mathbb{N}$, a many-to-one assignment $A$ of agents in $\mathcal{N}_1$ to agents in $\mathcal{N}_2$ is an \emph{$(\mathcal{N}_1, \mathcal{N}_2,k)$-sufficiently representative assignment} if:
\begin{itemize}
\item For every agent $j \in \mathcal{N}_2$, there are at most $\sqrt{n}$ agents from $\mathcal{N}_1$ assigned to $j$;
\item 
For any bipartite graph $H$ with node set $\mathcal{N}_1 \cup \mathcal{N}_2$ and maximum degree $k$, there are at most $k\sqrt{n}$ agents in  $\mathcal{N}_1$ that prefer some of their neighbors in $H$ to the agent they are assigned to according to $A$.
\end{itemize}
\end{definition}

For One-Sided Matching and $k$-Constrained Resource Allocation, $\mathcal{N}_1$ is the set of actual agents and $\mathcal{N}_2$ is the set of items. In contrast, for all other problems considered here we have $\mathcal{N}_1= \mathcal{N}_2 = \mathcal{N}$.

Like we did in Section \ref{sec:NIA}, we need to show that an $(\mathcal{N}_1, \mathcal{N}_2,k)$-sufficiently representative assignment exists for any instance of Ordinal-$k$-Max-on-Graphs and any $\mathcal{N}_1, \mathcal{N}_2 \subseteq \mathcal{N}$. We rely on the same high-level idea for the construction: $\sqrt{n}$ copies of each agent in $\mathcal{N}_2$ are created and then a {\sc Serial Dictatorship} algorithm is run with respect to the agents in $\mathcal{N}_1$.
The running time of {\sc $\sqrt{n}$-Serial Dictatorship} remains $O(n^{1.5})$. 

\begin{algorithm}[h!]
{\fontsize{10}{12}\selectfont
\begin{algorithmic}[1]
 \State 
 Let $\mathcal{B}$ be a multiset with $\sqrt{n}$ copies of each $j\in \mathcal{N}_2$ \;
  \For {\normalfont every agent $i \in \mathcal{N}_1$}
	   \State Let $\alpha_i$ be  $i$'s most preferred agent in $\mathcal{B}$ w.r.t.~$\succ_i$  \;
	   \State Remove $\alpha_i$ from $\mathcal{B}$ \;
   \EndFor
\State \Return $A = (\alpha_i)_{i \in \mathcal{N}_1}$ \;
\end{algorithmic}
}
\caption{\textsc{\sc $\sqrt{n}$-Serial Dictatorship}$(\mathcal{N}_1, \mathcal{N}_2, \sucv)$}
\label{alg:SD-sqrt_gen}
\end{algorithm}

The following theorem is the analog of Theorem \ref{thm:matching-GAGA}. While the proof is very similar, the counting argument here is somewhat less intuitive compared to the case of One-Sided Matching due to the differences between Definitions \ref{def:nia} and \ref{def:gaga}.

\begin{theorem}\label{thm:general-GAGA}
The assignment computed by the {\sc $\sqrt{n}$-Serial Dictatorship} algorithm is an $(\mathcal{N}_1, \mathcal{N}_2,k)$-sufficiently representative assignment.
\end{theorem}

\begin{proof}
Let $A$ be the assignment produced by the algorithm. During the execution of the algorithm, whenever all the copies of an agent $j\in \mathcal{N}_2$  have been matched, we will say that $j$ is \emph{exhausted}. Assume towards a contradiction that $A$ is not an $(\mathcal{N}_1, \mathcal{N}_2,k)$-sufficiently representative assignment. By construction, the first condition of Definition~\ref{def:gaga} is obviously satisfied. So, there must be a graph $H$, as described in the second condition of Definition \ref{def:gaga}, with respect to which there exists a subset $\mathcal{S}_1 \subseteq \mathcal{N}_1$ with $|\mathcal{S}_1|>k\sqrt{n}$ such that every $i\in \mathcal{S}_1$ prefers her best 
neighbor in $H$, say $\beta_i$, to agent $\alpha_i$ she has been assigned to in $A$. Let $\mathcal{S}_2 \subseteq \mathcal{N}_2$ be the set that contains all these $\beta_i$s.
Because $H$ has maximum 
degree at most $k$ we have $|\mathcal{S}_2|\ge |\mathcal{S}_1|/k > \sqrt{n}$.


Consider any agent $i \in \mathcal{S}_1$.
The fact that this agent was not assigned to $\beta_i$ by \textsc{$\sqrt{n}$-Serial Dictatorship} implies that when it was $i$'s turn to pick, agent $\beta_i$ was exhausted. 
Therefore, at the end of the algorithm, all agents of $\mathcal{S}_2$ must be exhausted. 
Since an agent in $\mathcal{N}_2$ is exhausted when all its $\sqrt{n}$ copies have been assigned and there are at most $n$ agents in $\mathcal{N}_1$, we can only have as many as $\frac{n}{\sqrt{n}} = \sqrt{n}$ exhausted agents. The latter means that $|\mathcal{S}_2| \leq \sqrt{n}$, a contradiction.
\end{proof}

\subsection{The General Mechanism}
We are now ready to show that it is possible to achieve distortion $O(\sqrt{n})$
for \emph{any} problem that can be modeled as a special case of Ordinal-$\bm{k}$-Max-on-Graphs. 
Our mechanism generalizes the main idea of \nameref{alg:Match-TwoQueries} of querying each agent about for her overall favorite alternative, as well as the alternative suggested by an appropriate sufficiently representative assignment. For the latter, we need to specify $\mathcal{N}_1$ and $\mathcal{N}_2$: These are typically both equal to the whole $\mathcal{N}$, unless the problem distinguishes between actual agents and items, in which case these two groups are captured by $\mathcal{N}_1$ and $\mathcal{N}_2$, respectively. In any case, all the edges in a feasible solution have at least one endpoint in each of  $\mathcal{N}_1$ and $\mathcal{N}_2$.

\begin{algorithm}[h!]
{\fontsize{10}{12}\selectfont
\begin{algorithmic}[1]
	\State Query each agent $i\in \mathcal{N}_1$ for her favorite alternative in $\mathcal{N}_2$ w.r.t.~$\succ_i$
	\State Compute $A = \!\sqrt{n}\textsc{\sc -Serial Dictatorship}{\small (\mathcal{N}_1, \mathcal{N}_2, \sucv\!)}$
	\State Query each agent $i\in \mathcal{N}_1$ about the agent $\alpha_i \in \mathcal{N}_2$ she is assigned to in $A$
	\State Set all non-revealed values to $0$
	\State \Return a maximum-weight member of $\mathcal{F}$
\end{algorithmic}
}
\caption{\textsc{\sc General-TwoQueries}{\small $(G, \mathcal{F}, \mathcal{N}_1, \mathcal{N}_2, \sucv)$}} \label{alg:Max-on-Graphs-TwoQueries}
\end{algorithm}

Note that the final step of the algorithm involves computing a solution that is optimal according to the revealed values. There are computational issues to consider here, however, we first tackle the question of whether it is even possible to match the lower bounds of Theorem \ref{thm:general-lower} for $\lambda=2$, despite the lack of information. We briefly discuss how to transform {\sc General-TwoQueries} into a polynomial-time mechanism after the proof of Theorem \ref{thm:two-queries-general} below. Again, for the statement of the theorem we assume that $k\in \mathbb{N}$ is a constant and that, for every $G$, a matching extending $k$-family $\mathcal{F}(G)$ is specified.

\begin{customthm}{4}
\label{thm:two-queries-general}
For Ordinal-${k}$-Max-on-Graphs with feasible solutions given by $\mathcal{F}(\,\cdot\,)$,  {\sc General-TwoQueries}  has distortion $O(\sqrt{n})$.
\end{customthm}
\vspace{-7pt}
\begin{proof}
Consider any instance with valuation profile $\bv$ and relevant sets of agents $\mathcal{N}_1$ and $\mathcal{N}_2$.
Let $Y$ be the solution computed by the {\sc General-TwoQueries} mechanism when given as input $(G, \mathcal{F}, \mathcal{N}_1, \mathcal{N}_2, \sucv)$, and let $X$ denote an optimal solution. Let $w_R(Y)$ be the \emph{revealed} weight of $Y$ as seen by the mechanism, that is, the weight of $Y$ taking into account only the values that have been revealed by the queries. We will show that $w(X) \leq (1 + 10k^2\sqrt{n} ) \cdot w_R(Y)$, and the bound on the distortion will then follow by the obvious fact that $w(Y) \geq w_R(Y)$.

We can write the optimal weight as
\begin{align}\label{eq:general-two-queries-X}
w(X) = w_R(X) + w_C(X),
\end{align}
where $w_R(X)$ is the 
revealed weight of $X$ that takes into account only the values that have been revealed by the queries,
whereas $w_C(X)$ is the \emph{concealed} weight of $X$ that takes into account only the values that have not been revealed by the queries of the mechanism.
Since $Y$ is the solution that maximizes the social welfare based only on the revealed values, we have that
\begin{align}\label{eq:general-two-queries-revealed}
w_R(X) \leq w_R(Y).
\end{align}
Thus, it suffices to bound $w_C(X)$.

Let $S$ be the set of agents in $\mathcal{N}_1$ who are not queried about all their 
neighbors in $X$.
We partition $S$ into two subsets $S^{\geq}$ and $S^{<}$ consisting of agents for whom the second query of the mechanism is used to ask about someone they consider {\em better} or {\em worse} than their best 
neighbor in $X$, respectively.
For an agent $i\in \mathcal{N}_1$, let $\chi_i$ be $i$'s favorite 
neighbor in $X$ and recall that $i$ is queried about agent $\alpha_i \in \mathcal{N}_2$ to whom she is assigned according to the $(\mathcal{N}_1, \mathcal{N}_2,k)$-sufficiently representative assignment $A$. So,
\begin{align*}
&S^{\geq} = \left\{ i \in S: v_{i,\alpha_i} \geq v_{i,\chi_i}\right\}, \\
&S^{<} = \left\{ i \in S: v_{i,\alpha_i} < v_{i,\chi_i}\right\}.
\end{align*}

Let $N_X(i)$ be the set of agents who are 
neighbors of $i$ in $X$ \emph{and} for which $i$ was not queried about. 
We now define
\begin{align*}
& w_C^{\geq}(X) = \sum_{i \in S^{\geq}} \sum_{j \in N_X(i)} v_{i,j} \\
& w_C^{<}(X) = \sum_{i \in S^{<}} \sum_{j \in N_X(i)} v_{i,j}
\end{align*}
Clearly, $w_C(X) = w_C^{\geq}(X) + w_C^{<}(X)$.

For every agent $j \in \mathcal{N}_2$, let $S_j^{\geq} = \{ i \in S^{\geq}: \alpha_i=j \}$ be the set of all agents in $S^{\geq}$ that are queried about $j$ by the mechanism using the second query. So, $S^{\geq} = \bigcup_{j \in \mathcal{N}_2}S_j^{\geq}$.
Since $A$ is an $(\mathcal{N}_1, \mathcal{N}_2,k)$-sufficiently representative assignment, the first condition of Definition \ref{def:gaga} implies that $|S_j^{\geq}| \leq \sqrt{n}$ for every $j\in \mathcal{N}_2$. Therefore,

\begin{align}\label{eq:general-two-queries-concealed-geq}
w_C^{\geq}(X)
&= \sum_{j \in \mathcal{N}_2} \sum_{i \in S_j^{\geq}} \sum_{\ell \in N_X(i)} v_{i,\ell} \nonumber\\
&\leq \sum_{j \in \mathcal{N}_2} \sum_{i \in S_j^{\geq}} \sum_{\ell \in N_X(i)} v_{i,j} \nonumber\\
&\leq \sum_{j \in \mathcal{N}_2} \sum_{i \in S_j^{\geq}} k \cdot v_{i,j} \nonumber \\
&\leq k \sum_{j \in \mathcal{N}_2} |S_j^{\geq}| \max_{i \in S_j^{\geq}}v_{i,j} \nonumber \\
&\leq k \sqrt{n} \sum_{j \in \mathcal{N}_2} \max_{i \in S_j^{\geq}}v_{i,j} 
\end{align}
where the first inequality holds by the definition of the sets $S_j^{\geq}$, for every $j \in \mathcal{N}_2$. To complete our bound on $w_C^{\geq}(X)$, we need the following claim.

\begin{claim}\label{claim:bound-sum-concealed-geq}
For all the problems of interest, 
\[\sum_{j \in \mathcal{N}_2} \max_{i \in S_j^{\geq}}v_{i,j} \le 9k\cdot w_R(Y).\]
\end{claim}

\begin{proof}
For each $j \in \mathcal{N}_2$, let $i_j \in \argmax_{i \in S_j^{\geq}}v_{i,j}$. Consider the subgraph $H$ of the input graph $G$ with edge set $E(H) = \{\{i_j, j\} \,|\, j \in \mathcal{N}_2\}\}$, i.e., $H$ contains exactly the edges that define the sum of interest. In particular, we have
\begin{align*}
\sum_{j \in \mathcal{N}_2} \max_{i \in S_j^{\geq}}v_{i,j} = \sum_{j \in \mathcal{N}_2} v_{i_j,j}   =  w_R(H) . 
\end{align*}
We now claim that each node in $H$ has degree at most $2$. To see this, consider an agent $\ell\in \mathcal{N}$. There is at most one  $j\in \mathcal{N}_2$ such that $\ell\in S_j^{\geq}$ (since these sets are disjoint), and thus we may have $\ell = i_j$ for at most one $j\in \mathcal{N}_2$, resulting in the edge $\{j, \ell\}$ in $H$.  Additionally, $\ell$ may itself be in $\mathcal{N}_2$, resulting in a second edge $\{\ell, i_{\ell}\}$ in $H$. Other than these two, there can be no other edges of $H$ adjacent to $\ell$.

Since $H$ has maximum degree at most $2$, it must contain a matching $M$ of comparable weight. Specifically, $H$ must be the union of node-disjoint paths and cycles. We construct a (possibly empty) matching $M_1$ on $H$ by arbitrarily picking one edge from each odd cycle and one edge from the beginning of each odd path. If we remove $M_1$ from $H$, then the remaining graph is the union of node-disjoint even paths and even cycles, and thus can be decomposed into two disjoint matchings $M_2, M_3$ in a straightforward way. Since $M_1$, $M_2$, and $M_3$ cover all the edges of $H$, the best of them, say $M$, must have weight at least $w_R(H) / 3$, i.e, 
\begin{align*}
w_R(M|\bv)   = \frac{1}{3} w_R(H) . 
\end{align*}
Now we can work with $M$ like in the proof of Theorem \ref{thm:general-lower}. Consider the submatching $\hat{M}$ of $M$ containing the $\lfloor 2n/3k \rfloor$ heaviest edges of $M$ to get
\begin{align} \label{eq:distortion-heavy-submatching}
w_R(\hat{M}) \ge  \frac{1}{3k} w_R(M) \ge \frac{1}{9k} w_R(H).
\end{align}
Since $\mathcal{F}$ is a matching extending $k$-family and $\hat{M}$ is sufficiently small, there is some $\hat{Y}\in \mathcal{F}(G)$ such that $\hat{M}$ is a subgraph of $\hat{Y}$.
As $Y$ is a maximum-weight element of $\mathcal{F}(G)$ with respect to the revealed weights,
we directly get $w_R(Y) \ge w_R(\hat{Y}) \ge w_R(\hat{M})$ and, combining with \eqref{eq:distortion-heavy-submatching}, we have 
\begin{align*} 
w_R(Y) \ge \frac{1}{9k} w(H),
\end{align*}
as desired.
\end{proof}
By combining \eqref{eq:general-two-queries-concealed-geq} with Claim \ref{claim:bound-sum-concealed-geq}, we get 
\begin{align}\label{eq:two-queries-concealed-geq-final}
w_C^{\geq}(X) \leq 9 k^2 \sqrt{n} \cdot w_R(Y).
\end{align}


\color{black}

We next consider the quantity $w_C^{<}(X)$. By the fact that $A$ is an $(\mathcal{N}_1, \mathcal{N}_2,k)$-sufficiently representative assignment, it follows that $|S^{<}| \leq k\sqrt{n}$; otherwise $X$ would be a graph that violates the second condition of Definition \ref{def:gaga}.
Combined with the fact that all agents in $\mathcal{N}_1$ are queried about their favorite alternative, we can obtain the following upper bound on
$w_C^{<}(X)$. 
Recall that for $i\in \mathcal{N}_1$, we have $N_X(i)\subseteq \mathcal{N}_2$ and $|N_X(i)|\le k$.
\begin{align} \label{eq:general-two-queries-concealed-<}
w_C^{<}(X)
&= \sum_{i \in S^{<}} \sum_{j \in N_X(i)} v_{i,j} \nonumber \\
&\leq \sum_{i \in S^{<}} k \cdot \max_{j \in \mathcal{N}} v_{i,j} \nonumber \\
&\leq k\, |S^{<}|  \, \max_{i \in S^{<}} \max_{j \in \mathcal{N}} v_{i,j} \nonumber \\
&\leq k^2 \sqrt{n} \cdot w_R(Y).
\end{align}

The bound now follows by \eqref{eq:general-two-queries-X}, \eqref{eq:general-two-queries-revealed}, \eqref{eq:two-queries-concealed-geq-final}, \eqref{eq:general-two-queries-concealed-<}.
\end{proof}




A subtle point here is that of computational efficiency. Although designing polynomial time mechanisms is not our primary goal, it is clear that the only possible bottleneck is the last step of {\sc General-TwoQueries}. 
Indeed, the mechanism runs in polynomial time whenever there is a polynomial-time algorithm (exact or $O(1)$-approximation) for the full information version of the corresponding optimization problem.
The good news are that all variants of matching problems we presented can be solved efficiently by Edmond's algorithm \citep{edmonds1965maximum} or its extensions \citep{marsh1979matching}.

\begin{corollary}
There are deterministic polynomial-time mechanisms for General Graph Matching, Two-Sided Matching, General Graph k-Matching, and k-Constrained Resource Allocation which all use at most two queries per agent and have distortion $O(\sqrt{n})$. 
\end{corollary}

\section{Towards Tight Bounds for General Social Choice} 
\label{sec:social-choice}

Here we consider the general social choice setting where a set $\mathcal{N}$ of $n$ {\em agents} have preferences over a set $\mathcal{A}$ of $m$ {\em alternatives}. As in the One-Sided Matching problem, there is a valuation profile $\boldv=(v_{i,j})_{i \in \mathcal{N}, j \in \mathcal{A}}$ specifying the non-negative value that each agent $i$ has for every alternative $j$. The goal is to choose a single alternative $x \in \mathcal{A}$ to maximize the social welfare, that is, the total value of the agents for $x$: $\SW(x| \boldv) = \sum_{i \in \mathcal{N}} v_{i,x}$. Again, when $\bv$ is clear from context, we will drop it from notation. Similarly to One-Sided Matching and the problems discussed in the previous section, $\bv$ is unknown, and we are only given access to the ordinal profile $\sucv$ that is induced by $\bv$. Social choice mechanisms must decide a single alternative based only on $\sucv$ and the values they can learn by making a small number of queries. The notion of distortion (Definition~\ref{def:distortion}) can be extended for this setting as well, by taking the supremum over all instances with $n$ agents and $m$ alternatives, and letting $\mathcal{X}$ be the set $\mathcal{A}$ of alternatives.

For this general social choice setting, \citet{amanatidis2020peeking} showed a lower bound of $\Omega\left(m^{1/(2(\lambda+1))}\right)$ on the distortion of mechanisms that make at most $\lambda \geq 1$ queries per agent. We improve this result by showing a lower bound of $\Omega(m^{1/\lambda})$ for any constant $\lambda$. 

\begin{theorem}\label{thm:sc-lower}
In the social choice setting, the distortion of any deterministic mechanism that makes at most a constant number $\lambda \geq 1$ of queries per agent is $\Omega(m^{1/\lambda})$.
\end{theorem}

\begin{proof}
Let $\mathcal{M}$ be an arbitrary mechanism that makes at most $\lambda \geq 1$ queries per agent. Consider the following instance with $n$ agents and $m=n$ alternatives. We assume that $m$ satisfies the condition $m \geq \frac12 \sum_{\ell=1}^\lambda m^{(\lambda-\ell+1)/\lambda} + 2$, and also that it is superconstant; otherwise the theorem holds trivially.  We partition the set of alternatives $\mathcal{A}$ into $\lambda+2$ sets $A_1$, $A_2$, ... $A_{\lambda+1}$, $A_{\lambda+2}$, such that
\begin{itemize}
    \item $|A_\ell| = \frac12 m^{(\lambda-\ell+1)/\lambda}$ for $\ell \in [\lambda]$; 
    \item $|A_{\lambda+1}| = 2$;
    \item $|A_{\lambda+2}| = m - \frac12 \sum_{\ell=1}^\lambda m^{(\lambda-\ell+1)/\lambda} - 2.$
\end{itemize}
The ordinal profile has the following properties:
\begin{itemize}
    \item For every $\ell \in [\lambda+1]$, each alternative $j \in A_\ell$ is ranked at position $\ell$ by a set $T_{j, \ell}$ of $\frac{m}{|A_\ell|} = \Theta\left(m^{(\ell-1)/\lambda}\right)$ agents.
    \item For every $\ell \in [\lambda]$, every pair of agents that rank the same alternative in $A_\ell$ at position $\ell$, rank the same alternative in $A_{\ell+1}$ at position $\ell+1$.  
    \item For every agent, the alternatives that she does not rank in the first $\lambda+1$ positions are ranked arbitrarily from position $\lambda+2$ to $m$.
\end{itemize}
An example of the ordinal profile when $\lambda=2$ is depicted in Figure~\ref{fig:lower-2-queries} (see supplementary material). For every agent $i$, a query of $\mathcal{M}$ for alternative $j$ reveals a value of 
\begin{itemize}
    \item $m^{-\ell/\lambda}$ if $i$ ranks $j$ at position $\ell \in [\lambda+1]$, and
    \item  and a value of $0$ if $i$ ranks $j$ at any other position. 
\end{itemize}

\begin{figure*}[t!]
\centering

\tikzset{every picture/.style={line width=0.75pt}} 

\begin{tikzpicture}[x=0.75pt,y=0.75pt,yscale=-1,xscale=1]

\draw  [fill={rgb, 255:red, 155; green, 155; blue, 155 }  ,fill opacity=0.8 ] (99,42) -- (130,42) -- (130,57) -- (99,57) -- cycle ;
\draw  [fill={rgb, 255:red, 155; green, 155; blue, 155 }  ,fill opacity=0.8 ] (192,45) -- (226,45) -- (226,99.5) -- (192,99.5) -- cycle ;
\draw  [fill={rgb, 255:red, 155; green, 155; blue, 155 }  ,fill opacity=0.8 ] (99,64) -- (130,64) -- (130,79) -- (99,79) -- cycle ;
\draw  [fill={rgb, 255:red, 155; green, 155; blue, 155 }  ,fill opacity=0.8 ] (99,86) -- (130,86) -- (130,101) -- (99,101) -- cycle ;
\draw  [fill={rgb, 255:red, 155; green, 155; blue, 155 }  ,fill opacity=0.8 ] (99,110) -- (130,110) -- (130,125) -- (99,125) -- cycle ;
\draw  [fill={rgb, 255:red, 155; green, 155; blue, 155 }  ,fill opacity=0.8 ] (99,134) -- (130,134) -- (130,149) -- (99,149) -- cycle ;
\draw  [fill={rgb, 255:red, 155; green, 155; blue, 155 }  ,fill opacity=0.8 ] (99,158) -- (130,158) -- (130,173) -- (99,173) -- cycle ;
\draw  [fill={rgb, 255:red, 155; green, 155; blue, 155 }  ,fill opacity=0.8 ] (99,182) -- (130,182) -- (130,197) -- (99,197) -- cycle ;
\draw  [fill={rgb, 255:red, 155; green, 155; blue, 155 }  ,fill opacity=0.8 ] (99,261) -- (130,261) -- (130,276) -- (99,276) -- cycle ;
\draw  [fill={rgb, 255:red, 155; green, 155; blue, 155 }  ,fill opacity=0.8 ] (99,285) -- (130,285) -- (130,300) -- (99,300) -- cycle ;
\draw  [fill={rgb, 255:red, 155; green, 155; blue, 155 }  ,fill opacity=0.8 ] (99,309) -- (130,309) -- (130,324) -- (99,324) -- cycle ;
\draw  [fill={rgb, 255:red, 155; green, 155; blue, 155 }  ,fill opacity=0.8 ] (192,133.5) -- (226,133.5) -- (226,188) -- (192,188) -- cycle ;
\draw  [fill={rgb, 255:red, 155; green, 155; blue, 155 }  ,fill opacity=0.8 ] (192,264) -- (226,264) -- (226,318.5) -- (192,318.5) -- cycle ;
\draw  [fill={rgb, 255:red, 155; green, 155; blue, 155 }  ,fill opacity=0.8 ] (289,40) -- (323,40) -- (323,186.5) -- (289,186.5) -- cycle ;
\draw  [fill={rgb, 255:red, 155; green, 155; blue, 155 }  ,fill opacity=0.8 ] (289,201) -- (323,201) -- (323,335.5) -- (289,335.5) -- cycle ;
\draw  [fill={rgb, 255:red, 155; green, 155; blue, 155 }  ,fill opacity=0.8 ] (379,41) -- (533,41) -- (533,329.5) -- (379,329.5) -- cycle ;
\draw  [fill={rgb, 255:red, 155; green, 155; blue, 155 }  ,fill opacity=0.8 ] (73,61.13) -- (80.5,39) -- (88,61.13) -- (84.25,61.13) -- (84.25,105.38) -- (88,105.38) -- (80.5,127.5) -- (73,105.38) -- (76.75,105.38) -- (76.75,61.13) -- cycle ;
\draw  [fill={rgb, 255:red, 155; green, 155; blue, 155 }  ,fill opacity=0.8 ] (99,206) -- (130,206) -- (130,221) -- (99,221) -- cycle ;
\draw  [fill={rgb, 255:red, 155; green, 155; blue, 155 }  ,fill opacity=0.8 ] (73,156.13) -- (80.5,134) -- (88,156.13) -- (84.25,156.13) -- (84.25,200.38) -- (88,200.38) -- (80.5,222.5) -- (73,200.38) -- (76.75,200.38) -- (76.75,156.13) -- cycle ;
\draw    (135,51.5) -- (186.02,59.2) ;
\draw [shift={(188,59.5)}, rotate = 188.58] [color={rgb, 255:red, 0; green, 0; blue, 0 }  ][line width=0.75]    (10.93,-3.29) .. controls (6.95,-1.4) and (3.31,-0.3) .. (0,0) .. controls (3.31,0.3) and (6.95,1.4) .. (10.93,3.29)   ;
\draw    (136,74) -- (186.03,65.82) ;
\draw [shift={(188,65.5)}, rotate = 530.72] [color={rgb, 255:red, 0; green, 0; blue, 0 }  ][line width=0.75]    (10.93,-3.29) .. controls (6.95,-1.4) and (3.31,-0.3) .. (0,0) .. controls (3.31,0.3) and (6.95,1.4) .. (10.93,3.29)   ;
\draw    (137,96.5) -- (187.18,73.34) ;
\draw [shift={(189,72.5)}, rotate = 515.22] [color={rgb, 255:red, 0; green, 0; blue, 0 }  ][line width=0.75]    (10.93,-3.29) .. controls (6.95,-1.4) and (3.31,-0.3) .. (0,0) .. controls (3.31,0.3) and (6.95,1.4) .. (10.93,3.29)   ;
\draw    (137,117) -- (187.32,84.58) ;
\draw [shift={(189,83.5)}, rotate = 507.21] [color={rgb, 255:red, 0; green, 0; blue, 0 }  ][line width=0.75]    (10.93,-3.29) .. controls (6.95,-1.4) and (3.31,-0.3) .. (0,0) .. controls (3.31,0.3) and (6.95,1.4) .. (10.93,3.29)   ;
\draw    (135,138.5) -- (186.02,146.2) ;
\draw [shift={(188,146.5)}, rotate = 188.58] [color={rgb, 255:red, 0; green, 0; blue, 0 }  ][line width=0.75]    (10.93,-3.29) .. controls (6.95,-1.4) and (3.31,-0.3) .. (0,0) .. controls (3.31,0.3) and (6.95,1.4) .. (10.93,3.29)   ;
\draw    (136,161) -- (186.03,152.82) ;
\draw [shift={(188,152.5)}, rotate = 530.72] [color={rgb, 255:red, 0; green, 0; blue, 0 }  ][line width=0.75]    (10.93,-3.29) .. controls (6.95,-1.4) and (3.31,-0.3) .. (0,0) .. controls (3.31,0.3) and (6.95,1.4) .. (10.93,3.29)   ;
\draw    (137,183.5) -- (187.18,160.34) ;
\draw [shift={(189,159.5)}, rotate = 515.22] [color={rgb, 255:red, 0; green, 0; blue, 0 }  ][line width=0.75]    (10.93,-3.29) .. controls (6.95,-1.4) and (3.31,-0.3) .. (0,0) .. controls (3.31,0.3) and (6.95,1.4) .. (10.93,3.29)   ;
\draw    (137,204) -- (187.32,171.58) ;
\draw [shift={(189,170.5)}, rotate = 507.21] [color={rgb, 255:red, 0; green, 0; blue, 0 }  ][line width=0.75]    (10.93,-3.29) .. controls (6.95,-1.4) and (3.31,-0.3) .. (0,0) .. controls (3.31,0.3) and (6.95,1.4) .. (10.93,3.29)   ;
\draw    (228,76.5) -- (277.28,105.49) ;
\draw [shift={(279,106.5)}, rotate = 210.47] [color={rgb, 255:red, 0; green, 0; blue, 0 }  ][line width=0.75]    (10.93,-3.29) .. controls (6.95,-1.4) and (3.31,-0.3) .. (0,0) .. controls (3.31,0.3) and (6.95,1.4) .. (10.93,3.29)   ;
\draw    (235,162.5) -- (277.63,116.96) ;
\draw [shift={(279,115.5)}, rotate = 493.11] [color={rgb, 255:red, 0; green, 0; blue, 0 }  ][line width=0.75]    (10.93,-3.29) .. controls (6.95,-1.4) and (3.31,-0.3) .. (0,0) .. controls (3.31,0.3) and (6.95,1.4) .. (10.93,3.29)   ;

\draw (195,3.4) node [anchor=north west][inner sep=0.75pt]   {${\displaystyle \frac{\sqrt{m}}{2}}$};
\draw (108,3.4) node [anchor=north west][inner sep=0.75pt]   {${\displaystyle \frac{m}{2}}$};
\draw (112,225) node [anchor=north west][inner sep=0.75pt]    {$\vdots $};
\draw (204,210) node [anchor=north west][inner sep=0.75pt]    {$\vdots $};
\draw (301,13.4) node [anchor=north west][inner sep=0.75pt]  {${\displaystyle 2}$};
\draw (87,343) node [anchor=north west][inner sep=0.75pt]   [align=left] {{position $1$}};
\draw (183,343) node [anchor=north west][inner sep=0.75pt]   [align=left] {{position $2$}};
\draw (280,343) node [anchor=north west][inner sep=0.75pt]   [align=left] {{position $3$}};
\draw (403,343) node [anchor=north west][inner sep=0.75pt]   [align=left] {All other positions};
\draw (426,175) node [anchor=north west][inner sep=0.75pt]   [align=left] {arbitrary};
\draw (29,76.4) node [anchor=north west][inner sep=0.75pt]   {$2\sqrt{m}$};
\draw (29,170.4) node [anchor=north west][inner sep=0.75pt]   {$2\sqrt{m}$};
\end{tikzpicture}
\caption{An overview of the instance used in the proof of Theorem~\ref{thm:sc-lower} two queries ($\lambda=2$). Each rectangle in the first three positions corresponds to an alternative. Each rectangle at position $1$ contains two agents. Each rectangle at position $2$ contains the agents from $2\sqrt{m}$ rectangles at position $1$, as indicated by the arrows, meaning that those agents rank the same alternative second. The rectangles at position $3$ contain $m/2$ agents each, corresponding to $\sqrt{m}/4$ rectangles at position $2$. That is, the agents that rank second one of the first $\sqrt{m}/4$ alternatives at position $2$, rank third the alternative corresponding to the first rectangle at position $3$; similarly, the agents that rank second one of the last $\sqrt{m}/4$ alternatives at position $2$, rank third the alternative corresponding to the second rectangle at position $3$. The ranking of the alternatives in the remaining positions is consistent but otherwise arbitrary.}
\label{fig:lower-2-queries}
\end{figure*}
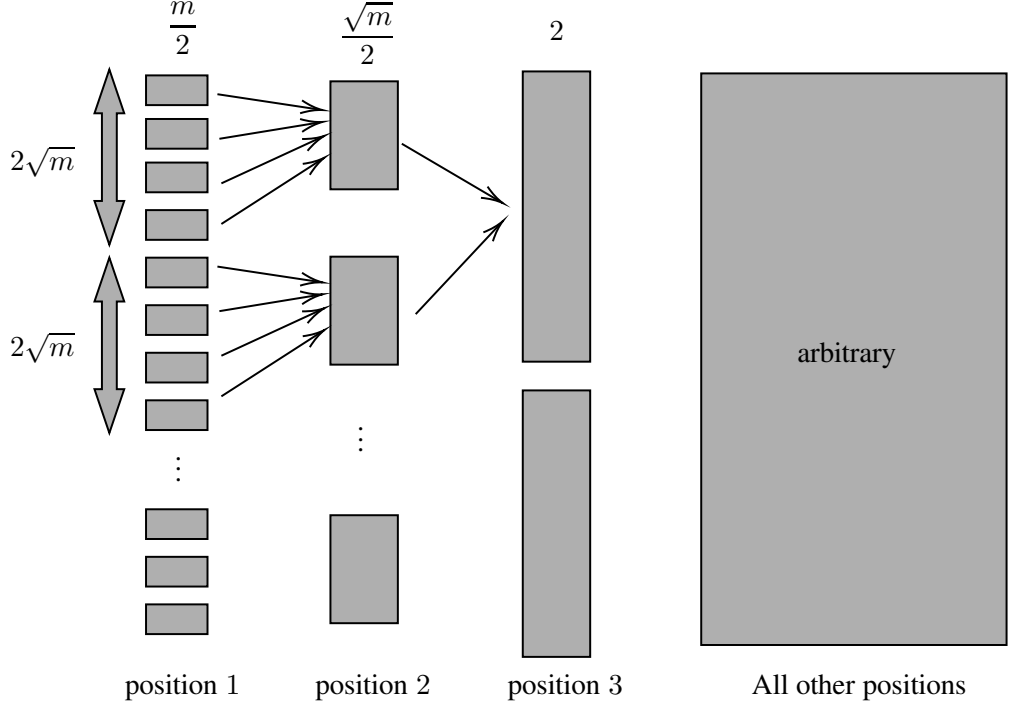

Given this instance as input, let $y$ be the alternative that $\mathcal{M}$ chooses as the winner. No matter the choice of $y$, we will define the cardinal profile so that it is consistent to the information revealed by the queries of $\mathcal{M}$, and the values of the agents for alternative $y$ are also consistent to the information that would have been revealed, irrespective of whether those values have actually been revealed. That is, any agent has a value of $m^{-\ell/\lambda}$ for $y$ if she ranks $y$ at position $\ell \in [\lambda+1]$, and a value of $0$ if she ranks $y$ at any other position. 
Hence, the social welfare of $y$ is 
\begin{itemize}
\item $\Theta\left(m^{(\ell-1)/\lambda}\right) \cdot m^{-\ell/\lambda} = \Theta(m^{-1/\lambda})$ if $y \in A_{\ell}$ for $\ell \in [\lambda+1]$, or
\item $0$ if $y \in A_{\lambda+2}$.
\end{itemize}
Consequently, to show the desired bound of $\Omega(m^{1/\lambda})$ on the distortion of $\mathcal{M}$, it suffices to assume that $y \in A_{\ell}$ for some $\ell \in [\lambda+1]$, and prove that the values of the agents that have not been revealed and do not correspond to alternative $y$ can always be defined such that there exists an alternative $x \neq y$ with social welfare $\Omega(1)$. 

Suppose towards a contradiction that the cardinal profile cannot be defined in a way so that there exists an alternative $x$ with social welfare $\Omega(1)$.
We make the following two observations:
\begin{itemize}
\item[(O1)] If there exists an alternative $x \in A_1 \setminus \{y\}$ for which at least one agent in $T_{x,1}$ is {\em not} queried by $\mathcal{M}$ for $x$, then we can set the value of this agent for $x$ to be constant. Consequently, all the agents in $\bigcup_{j \in A_1 \setminus \{y\}} T_{j,1}$ that rank alternatives different that $y$ at position $1$, must be queried at position $1$. 

\item[(O2)] Let $\varepsilon \in (0,1)$ be a constant and $\ell \in \{2, \dots, \lambda+1\}$. Consider any alternative $x \in A_\ell \setminus \{y\}$ and any set of agents $S \subseteq T_{x,\ell}$ such that $|S| \geq \varepsilon \cdot \frac{m}{|A_\ell|} = \Theta(m^{(\ell-1)/\lambda})$. If at least $\frac12  |S|$ agents in $S$ are {\em not} queried by $\mathcal{M}$ for $x$, then we could set the value of all these agents for $x$ to be $m^{-(\ell-1)/\lambda}$ (which is the revealed value when $\mathcal{M}$ queries for alternatives ranked at position $\ell-1$), and the social welfare of $x$ would be at least 
\[\frac12 |S| \cdot m^{-(\ell-1)/\lambda} = \Theta(m^{(\ell-1)/\lambda}) \cdot m^{-(\ell-1)/\lambda} = \Theta(1).\]
Consequently, for every alternative $x \in A_\ell \setminus \{y\}$ and set $S \subseteq T_{x,\ell}$ such that $|S| \geq \varepsilon \cdot \frac{m}{|A_\ell|}$, at least $\frac12 |S|$ agents in $S$ must be queried at position $\ell$ for $x$.  
\end{itemize}

Given these two observations, we are now ready to show by induction that the mechanism must make $\lambda+1$ queries for a high proportion of the agents, contradicting that $\mathcal{M}$ makes at most $\lambda$ queries per agent. 

For the base case, consider an alternative $x \in A_2 \setminus \{y\}$. By the definition of the ordinal profile, the agents in $T_{x,2}$ who rank $x$ at position $2$ are partitioned into $\frac{|A_1|}{|A_2|}$ subsets such that all $\frac{m}{|A_1|}$ agents in each subset rank first the same alternative of $A_1$. By (O1) we have that, besides the agents that rank alternative $y$ at position $1$, all other agents must be queried at position $1$. Hence, there exists a set $S \subseteq T_{x,2}$ consisting of 
$|S| \geq \left(\frac{|A_1|}{|A_2|}-1\right) \cdot \frac{m}{|A_1|}$ 
agents that are queried at position $1$. 
By the definitions of $A_1$ and $A_2$, and since $m$ is superconstant, we have that
$|S| \geq \frac12 \cdot \frac{m}{|A_2|}$. 
By (O2) for $\varepsilon = \frac12$ and $\ell = 2$, we have that at least $\frac12 |S| \geq \frac14 \cdot \frac{m}{|A_2|}$ of the agents in $S$ must also be queried at position $2$ for $x$. 

Let $\ell \in \{3, \dots, \lambda+1\}$ and assume as induction hypothesis that for every alternative $z \in A_{\ell-1} \setminus \{y\}$ there is a set of agents $S_z \subseteq T_{z,\ell-1}$ such that $|S_z| \geq \frac{1}{2^{2(\ell-2)}} \cdot \frac{m}{|A_{\ell-1}|}$ who are queried by $\mathcal{M}$ at the first $\ell-1$ positions. Consider an alternative $x \in A_\ell \setminus \{y\}$. By the definition of the ordinal profile, the agents in $T_{x,\ell}$ who rank alternative $x$ at position $\ell$ are partitioned into $\frac{|A_{\ell-1}|}{|A_\ell|}$ subsets such that all $\frac{m}{|A_{\ell-1}|}$ agents in each subset rank the same alternative in $A_{\ell-1}$ at position $\ell-1$. So, by our induction hypothesis, there is a set $S \subseteq T_{x,\ell}$ consisting of
\[|S| \geq \left( \frac{|A_{\ell-1}|}{|A_\ell|}-1 \right) \cdot \frac{1}{2^{2(\ell-2)}} \ \cdot \frac{m}{|A_{\ell-1}|}\]
agents that are queried at the first $\ell-1$ positions. 
By the definition of $A_{\ell-1}$ and $A_\ell$, and since $m$ is superconstant, we have that 
\[|S|  \geq \frac{1}{2^{2(\ell-2)+1}} \ \cdot \frac{m}{|A_\ell|}.\]
Since $\ell \leq \lambda+1$ and $\lambda$ is a constant, by observation (O2) for $\varepsilon = \frac{1}{2^{2(\ell-2)+1}}$, we have that at least 
\[\frac12 |S| \geq \frac{1}{2^{2(\ell-1)}}\cdot \frac{m}{|A_\ell|}\]
agents in $S$ must also be queried at position $\ell$ for $x$. 

Now, let $x \in A_{k+1} \setminus \{y\}$. The above induction shows that there are at least $\frac{1}{2^{2\lambda}}\cdot \frac{m}{|A_{\lambda+1}|}$ agents in $T_{x, \lambda+1}$ who must be queried by $\mathcal{M}$ at the first $\lambda+1$ positions. This contradicts the fact that $\mathcal{M}$ can make at most $k$ queries per agent, and the theorem follows.
\end{proof}

Our approach for all the problems discussed in the previous sections can also be applied to the much more general social choice setting, {\em subject to} being able to compute a particular set of alternatives.

\begin{definition} \label{def:sc-property}
Let $c \geq 1$ be any constant.
A subset of alternatives $B\subseteq \mathcal{A}$ with $|B| \leq c\cdot \sqrt{m}$ is a \emph{sufficiently representative set} if, for every alternative $j\in \mathcal{A}$, at most $\sqrt{m}$ agents prefer $j$ over their favorite alternative in $B$.
\end{definition}

We now present a mechanism that works {\em under the assumption} that sufficiently representative sets of alternatives can be (efficiently) computed; we discuss this assumption right after the statement of Theorem~\ref{thm:two-queries-social-choice}.

\begin{algorithm}[h!]
{\fontsize{10}{12}\selectfont
\begin{algorithmic}[1]
	\State Query each agent about her favorite alternative
	\State Compute a sufficiently representative set $B$
	\State Query each agent for her favorite alternative in $B$
	\State For every $j \in \mathcal{A}$, compute the revealed welfare $\SW_R(j)$
	\State \Return $y \in \arg\max_{j \in \mathcal{A}} \SW_R(j)$ 
\end{algorithmic}
}
\caption{\textsc{\sc SC-TwoQueries}$(\mathcal{N}, \mathcal{A}, \sucv)$} \label{alg:SC-TwoQueries}
\end{algorithm} 

In particular, \textsc{\sc SC-TwoQueries} (Mechanism \ref{alg:SC-TwoQueries}) first queries each agent about her overall favorite alternative (the one ranked first). Then, given a sufficiently representative set of alternatives $B$, it queries each agent for her favorite alternative in $B$. Given the answers to these two queries per agent, the mechanism outputs an alternative that maximizes the {\em revealed} social welfare which is based only on the values learned from the queries.

\begin{theorem}\label{thm:two-queries-social-choice}
The mechanism {\sc SC-TwoQueries} has distortion $O(\sqrt{m})$, when restricted to the social choice instances for which a sufficiently representative set of alternatives exists.
\end{theorem}

\begin{proof}
Consider any social choice instance with valuation profile $\bv$ that induces the ordinal preference profile $\sucv$. Let $y$ be the alternative chosen by the mechanism when given as input this instance, and denote by $x$ the optimal alternative. We will show that $\SW(x) \leq (1 + (1+c)\cdot \sqrt{m}) \SW_R(y)$. The bound on the distortion will then follow by the obvious fact that $\SW(y) \geq \SW_R(y)$.
	
We can write the optimal welfare as
\begin{align}\label{eq:sc-main-expression}
\SW(x) &= \SW_R(x) + \SW_C(x) \nonumber \\
&\le \SW_R(y)  + \SW_C(x) \,,
\end{align}
where $\SW_C(x)$ is the concealed welfare of $x$, consisting of the values of agents for $x$ that were not revealed by the queries of the mechanism, and the inequality follows by the fact that $y$ is the alternative that maximizes the revealed welfare. Let $S$ be the set of agents who were {\em not} queried about their value for $x$, and partition $S$ into the following two subsets: 
\begin{itemize}
    \item $S^{\geq}$ consists of the agents in $S$ for whom the second query is about an alternative that the agent considers {\em better} than $x$;
    \item $S^<$ consists of the agents in $S$ for whom the second query is about an alternative that the agent considers {\em worse} than $x$.
\end{itemize}
Given these sets, now let
\begin{align*}
\SW_{C}^{\geq}(x) = \sum_{i \in S^{\geq}}v_{i,x} \text{ \ \ \ and \ \ \ } \SW_C^{<}(x) = \sum_{i \in S^<}v_{i,x}.
\end{align*}
be the contribution of the agents in $S^{\geq}$ and of the agents in $S^<$ to the concealed welfare of $x$, respectively. 
That is, 
\[\SW_C(x) = \SW_C^{\geq}(x) + \SW_C^{<}(x).\]

By the definition of the mechanism, each agent is queried about her favorite alternative in the sufficiently representative set $B$. 
For every $j \in B \setminus \{x\}$, let $S^{\geq}_j \subseteq S^{\geq}$ be the set of agents in $S^{\geq}$ who are queried for alternative $j$ instead of $x$. Thus, $S^{\geq} = \bigcup_{j \in B \setminus \{x\}} S^{\geq}_j$. By the definition of $S^{\geq}$, the fact that $y$ maximizes the revealed welfare, and since $|B| \leq c\cdot \sqrt{m}$, we obtain  
\begin{align} 
\SW_{C}^{\geq}(x) 
&= \sum_{j \in B \setminus \{x\}} \sum_{i \in S^{\geq}_j} v_{i,x} \nonumber \\
&\leq \sum_{j \in B \setminus \{x\}}  \sum_{i\in S_j^{\geq}} v_{i,j}  \nonumber \\
&\leq \sum_{j \in B \setminus \{x\} } \SW_R(j|\bv) \nonumber \\
&\leq |B| \cdot \SW_R(y|\bv) \nonumber \\
&\leq c\cdot \sqrt{m} \cdot \SW_R(y). \label{eq:sc-concealed-geq}
\end{align}
	
Since all the agents in $S^{<}$ are queried for alternatives in the sufficiently representative set $B$ that they consider worse than $x$ and $B$, it must be the case that $|S^{<}| \leq \sqrt{m}$. Since all agents are queried at the first position for their favorite alternative, we obtain 
\begin{align}
\SW_{C}^{<}(x) 
&= \sum_{i \in S^{<}} v_{i,x}  \nonumber \\
&\leq \sum_{i \in S^{<}} \max_{j \in \mathcal{A}} v_{i,j} \nonumber \\
&\leq |S^{<}| \cdot \max_{i \in S^{<}} \max_{j \in \mathcal{A}} v_{i,j} \nonumber \\
&\leq \sqrt{m} \cdot \SW_R(y). \label{eq:sc-concealed-<}
\end{align}
	
The bound now follows by \eqref{eq:sc-main-expression}, \eqref{eq:sc-concealed-geq} and \eqref{eq:sc-concealed-<}.
\end{proof}

A sufficiently representative set of alternatives trivially exists when $m$ is much larger than $n$ (namely, when $m=\Omega(n^2)$). In contrast, when $m$ is much smaller than $n$, sufficiently representative sets of alternatives do not always exist.\footnote{For example, for any $k > \sqrt{m}$, consider an instance with $n = k \cdot m!$ agents, such that for each possible ordering of the $m$ alternatives there are exactly $k$ agents that have it as their preference. Then, for any subset $B$ of at most $\sqrt{m}$ alternatives and any alternative $j \in \mathcal{A} \setminus B$, there are at least $k > \sqrt{m}$ agents that prefer $j$ over any alternative in $B$.} 
\citet{jiang2020approximately} showed the following useful result:

\begin{theorem}[\citep{jiang2020approximately}]\label{thm:stable_committee}
For any $\xi \in [n]$, there exists a set $S$ of alternatives with $|S| \leq 16\cdot n/\xi$ such that for every $j \in A$, there are at most $\xi$ agents that prefer $j$ over their favorite alternative in $S$.
\end{theorem}

A set $S$ as in the theorem above is called an {\em approximately stable committee}~\cite{cheng2020group,jiang2020approximately}. Clearly, when $m = \Omega(n)$ and $\xi = \sqrt{n}$, an approximately stable committee is also a sufficiently representative set with $c = 16$. Therefore, combining  Theorems~\ref{thm:two-queries-social-choice} and \ref{thm:stable_committee}, we obtain the following.

\begin{corollary}
When $m = \Omega(n)$, {\sc SC-TwoQueries} has distortion $O(\sqrt{m})$.
\end{corollary}


\section{Conclusion and Open Problems} \label{sec:open}
In this paper, we showed that for a large class of problems, which includes One-Sided Matching and many other well-studied graph-theoretic problems, it is possible to achieve a distortion of $O(\sqrt{n})$ using a deterministic mechanism that makes at most two queries per agent, and that this is best possible asymptotically. Our whole methodology is based on computing assignments of agents to items or other agents that exhibit a very particular structure. In addition, in the social choice setting, when $m = \Omega(n)$, sets of alternatives with analogous properties can be computed, and our methodology yields a two-query mechanism with best possible distortion for this setting as well. 

It is an interesting open problem to design a mechanism that makes two queries and achieves the best possible distortion of $O(\sqrt{m})$ when $m = o(n)$, or show that this is impossible. We suspect that to obtain a positive result one would need to come up with an adaptive mechanism, which decides where to ask each query based not only on the ordinal information, but also on the answers to all previous ones.
Another question, about any of the problems we considered, is whether one can design mechanisms that make at most a constant $\lambda \geq 3$ queries per agent and their distortion matches the lower bound of $\Omega(n^{1/\lambda})$ (or, in the case of social choice, $\Omega(m^{1/\lambda})$). Again, we strongly suspect that the same type of adaptivity will be required for this task as well.

\section*{Acknowledgements}
This work is partially supported by the ERC Advanced Grant 788893 AMDROMA ``Algorithmic and Mechanism Design Research in Online Markets'', the MIUR PRIN project ALGADIMAR ``Algorithms, Games, 
 and Digital Markets'', and the NWO Veni project No.~VI.Veni.192.153.

\bibliographystyle{plainnat}
\bibliography{references}

\end{document}